\documentclass{article}
\usepackage[a4paper,portrait,margin=1in]{geometry}
\usepackage{dsfont}
\usepackage{complexity}
\usepackage{amsmath}
\usepackage{lmodern}
\usepackage{amssymb}
\usepackage{float}
\usepackage{graphicx}
\usepackage{tikz}
\usepackage{xcolor}
\usepackage{url}
\usepackage{paralist, enumerate}
\usepackage{refcount}
\DeclareMathOperator{\rank}{rank}

\newcommand{\megal}{M_\text{egal}}
\newcommand{\mopt}{M_\text{opt}}

\let\proof\relax

\usepackage{amsthm}
\newtheorem{theorem}{Theorem}
\newtheorem{pr}{Problem}
\newtheorem{cl}[theorem]{Claim}
\newtheorem{cor}[theorem]{Corollary}
\newtheorem{prop}[theorem]{Proposition}
\newcommand{\inp}{\textsf{Input: }} 
\newcommand{\outp}{\textsf{Output: }}
\newcommand{\ques}{\textsf{Question: }}
\newcommand{\true}{\mathsf{true}}
\newcommand{\false}{\mathsf{false}}
\usetikzlibrary{arrows,decorations.markings,patterns,calc, positioning, backgrounds}
\tikzstyle{vertex} = [circle, draw=black, scale=0.7]
\tikzstyle{edgelabel} = [circle, fill=white, scale=0.8]

\usepackage{algpseudocode}
\usepackage[ruled]{algorithm} 

\begin{document}

\title{\bf The Stable Roommates problem with short lists\thanks{The authors were supported by the Hungarian Academy of Sciences under its Momentum Programme (LP2016-3/2016), its J\'anos Bolyai Research Fellowship, OTKA grant K108383, COST Action IC1205 on Computational Social Choice and by EPSRC grant EP/K010042/1.} \thanks{A preliminary version of this paper appeared in the Proceedings of SAGT 2016: the 9th International Symposium on Algorithmic Game Theory.}
}

\author{\'Agnes Cseh$^1$,
        Robert W.\ Irving$^2$ and 
        David F.\ Manlove$^2$\\
        \\
   \small $^1$ \emph{Institute of Economics, Hungarian Academy of Sciences, and} \\
   \small \emph{Department of Operations Research and Actuarial Sciences, Corvinus University of Budapest, Hungary}\\
   \small \emph{Email:} {\tt cseh.agnes@krtk.mta.hu}\\ \\
   \small $^2$ \emph{School of Computing Science, University of Glasgow, UK}\\
  \small \emph{Email:} \{{\tt Rob.Irving,David.Manlove}\}{\tt @glasgow.ac.uk}
}
\date{ }
\maketitle

\begin{abstract}
We consider two variants of the classical Stable Roommates problem with Incomplete (but strictly ordered) preference lists ({\sc sri}) that are degree constrained, i.e., preference lists are of bounded length. The first variant, {\sc egal $d$-sri}, involves finding an egalitarian stable matching in solvable instances of {\sc sri} with preference lists of length at most~$d$. We show that this problem is $\NP$-hard even if $d=3$. On the positive side we give a $\frac{2d+3}{7}$-approximation algorithm for $d\in \{3,4,5\}$ which improves on the known bound of 2 for the unbounded preference list case. In the second variant of {\sc sri}, called {\sc$ d$-srti}, preference lists can include ties and are of length at most~$d$.  We show that the problem of deciding whether an instance of {\sc $d$-srti} admits a stable matching is $\NP$-complete even if $d=3$.  We also consider the ``most stable'' version of this problem and prove a strong inapproximability bound for the $d=3$ case.  However for $d=2$ we show that the latter problem can be solved in polynomial time.
\end{abstract}

\section{Introduction}

In the \emph{Stable Roommates problem with Incomplete lists}  (\textsc{sri}), a graph $G=(A,E)$ and a set of preference lists $\mathcal O$ are given, where the vertices $A=\{a_1,\dots,a_n\}$ correspond to \emph{agents}, and $\mathcal O=\{\prec_1,\dots,\prec_n\}$, where $\prec_i$ is a linear order on the vertices adjacent to $a_i$ in $G$ ($1\leq i\leq n$).  We refer to $\prec_i$ as $a_i$'s \emph{preference list}.  The agents that are adjacent to $a_i$ in $G$ are said to be \emph{acceptable} to~$a_i$.  If $a_j$ and $a_k$ are two acceptable agents for $a_i$ where $a_j\prec_i a_k$ then we say that $a_i$ \emph{prefers} $a_j$ to~$a_k$.

Let $M$ be a matching in~$G$. If $a_i a_j\in M$ then we let $M(a_i)$ denote~$a_j$. An edge $a_i a_j \notin M$ \emph{blocks} $M$, or forms a \emph{blocking edge} of $M$, if $a_i$ is unmatched or prefers $a_j$ to $M(a_i)$, and similarly $a_j$ is unmatched or prefers $a_i$ to~$M(a_j)$. A matching is called \emph{stable} if no edge blocks it. Denote by {\sc sr} the special case of {\sc sri} in which $G=K_n$.  Gale and Shapley~\cite{GS62} observed that an instance of {\sc sr} need not admit a stable matching. Irving~\cite{Irv85} gave a linear-time algorithm to find a stable matching or report that none exists, given an instance of {\sc sr}. The straightforward modification of this algorithm to the {\sc sri} case is described in~\cite{GI89}.  We call an {\sc sri} instance \emph{solvable} if it admits a stable matching.

In practice agents may find it difficult to rank a large number of alternatives in strict order of preference.  One natural assumption, therefore, is that preference lists are short, which corresponds to the graph being of bounded degree. Given an integer $d\geq 1$, we define {\sc $d$-sri} to be the restriction of {\sc sri} in which $G$ is of bounded degree~$d$. This special case of {\sc sri} problem has potential applications in organising tournaments. As already pointed out in a paper of Kujansuu et al.~\cite{KLM99}, {\sc sri} can model a pairing process similar to the Swiss system, which is used in large-scale chess competitions. The assumption on short lists is reasonable, because according to the Swiss system, players can be matched only to other players with approximately the same score.

A second variant of {\sc sri}, which can be motivated in a similar fashion, arises if we allow ties in the preference lists, i.e., $\prec_i$ ($1\leq i\leq n$) is now a strict weak ordering. That is, $\prec_i$ is a strict partial order in which incomparability is transitive. We refer to this problem as the \emph{Stable Roommates problem with Ties and Incomplete lists} (\textsc{srti})~\cite{IM02}. As in the {\sc sri} case, define {\sc $d$-srti} to be the restriction of {\sc srti} in which $G$ is of bounded degree~$d$.  Denote by {\sc srt} the special case of {\sc srti} in which $G=K_n$. In the context of the motivating application of chess tournament construction as mentioned in the previous paragraph, {\sc $d$-srti} is naturally obtained if a chess player has several potential partners of the same score and match history in the tournament.

In the {\sc srti} context, ties correspond to indifference in the preference lists. In particular, if $a_i a_j\in E$ and $a_i a_k\in E$ where $a_j\not\prec_i a_k$ and $a_k\not\prec_i a_j$ then $a_i$ is said to be \emph{indifferent between} $a_j$ and~$a_k$.  Thus preference in the {\sc sri} context corresponds to strict preference in the case of {\sc srti}.  Relative to the strict weak orders in $\mathcal O$, we can define stability in {\sc srti} instances in exactly the same way as for {\sc sri}. This means, for example, that if $a_i a_j\in M$ for some matching $M$, and $a_i$ is indifferent between $a_j$ and some agent $a_k$, then $a_i a_k$ cannot block~$M$. The term \emph{solvable} can be defined in the {\sc srti} context in an analogous fashion to {\sc sri}. Using a highly technical reduction from a restriction of \textsc{3-sat}, Ronn~\cite{Ron90} proved that the problem of deciding whether a given {\sc srt} instance is solvable is $\NP$-complete. A simpler reduction was given by Irving and Manlove~\cite{IM02}.  

For solvable instances of \textsc{sri} there can be many stable matchings.  Often it is beneficial to work with a stable matching that is fair to all agents in a precise sense~\cite{Gus87,ILG87}. One such fairness concept can be defined as follows. Given two agents $a_i$, $a_j$ in an instance $\mathcal{I}$ of {\sc sri}, where $a_i a_j\in E$, let $\rank(a_i,a_j)$ denote the rank of $a_j$ in $a_i$'s preference list (that is, 1 plus the number of agents that $a_i$ prefers to $a_j$). Let $A_M$ denote the set of agents who are matched in a given stable matching~$M$. (Note that this set depends only on $\mathcal I$ and is independent of $M$ by \cite[Theorem 4.5.2]{GI89}.) Define $c(M)=\sum_{a_i\in A_M} \rank(a_i,M(a_i))$ to be the \emph{cost} of $M$. An \emph{egalitarian stable matching} is a stable matching $M$ that minimises $c(M)$ over the set of stable matchings in~$\mathcal{I}$. Finding an egalitarian stable matching in \textsc{sr} was shown to be $\NP$-hard by Feder~\cite{Fed92}.  Feder~\cite{Fed92,Fed94} also gave a 2-approximation algorithm for this problem in the {\sc sri} setting. He also showed that an egalitarian stable matching in \textsc{sr} can be approximated within a factor of $\alpha$ of the optimum if and only if Minimum Vertex Cover can be approximated within the same factor~$\alpha$. It was proved later that, assuming the Unique Games Conjecture, Minimum Vertex Cover cannot be approximated within $2 - \varepsilon$ for any $\varepsilon > 0$~\cite{KR08}.

Given an unsolvable instance $\mathcal I$ of {\sc sri} or {\sc srti}, a natural approximation to a stable matching is a \emph{most-stable} matching~\cite{ABM06}.  Relative to a matching $M$ in $\mathcal I$, define $bp(M)$ to be the set of blocking edges of $M$ and let $bp(\mathcal I)$ denote the minimum value of $|bp(M')|$, taken over all matchings $M'$ in~$\mathcal I$.  Then $M$ is a \emph{most-stable} matching in $\mathcal I$ if $|bp(M)|=bp(\mathcal I)$.  The problem of finding a most-stable matching was shown to be $\NP$-hard and not approximable within $n^{k-\varepsilon}$, for any $\varepsilon>0$, unless $\P=\NP$, where $k=\frac{1}{2}$ if $\mathcal I$ is an instance of {\sc sr} and $k=1$ if $\mathcal I$ is an instance of {\sc srt} \cite{ABM06}.

To the best of our knowledge, there has not been any previous work published on either the problem of finding an egalitarian stable matching in a solvable instance of \textsc{sri} with bounded-length preference lists or the solvability of \textsc{srti} with bounded-length preference lists.  This paper provides contributions in both of these directions, focusing on instances of {\sc $d$-sri} and {\sc $d$-srti} for $d\geq 2$, with the aim of drawing the line between polynomial-time solvability and $\NP$-hardness for the associated problems in terms of~$d$.

\paragraph{Our contribution.} In Section~\ref{se:egal} we study the problem of finding an egalitarian stable matching in an instance of {\sc $d$-sri}.  We show that this problem is $\NP$-hard if $d=3$, whilst there is a straightforward algorithm for the case that $d=2$.  We then consider the approximability of this problem for the case that $d\geq 3$.  We give an approximation algorithm with a performance guarantee of $\frac{9}{7}$ for the case that $d=3$, $\frac{11}{7}$ if $d=4$ and $\frac{13}{7}$ if $d=5$. These performance guarantees improve on Feder's 2-approximation algorithm for the general {\sc sri} case~\cite{Fed92,Fed94}. In Section~\ref{se:srt} we turn to {\sc $d$-srti} and prove that the problem of deciding whether an instance of {\sc $3$-srti} is solvable is $\NP$-complete. We then show that the problem of finding a most-stable matching in an instance of {\sc $d$-srti} is solvable in polynomial time if $d=2$, whilst for $d=3$ we show that this problem is $\NP$-hard and not approximable within $n^{1-\varepsilon}$, for any $\varepsilon>0$, unless $\P=\NP$. Due to various complications, as explained in Section~\ref{se:no_egal_srti}, we do not attempt to define and study egalitarian stable matchings in instances of {\sc srti}.  Some open problems are presented in Section \ref{sec:open}.  A structured overview of previous results and our results (marked by~$\ast$) for \textsc{$d$-sri} and \textsc{$d$-srti} is contained in Table~\ref{ta:results}.
\begin{table}[t]
	\centering
    \resizebox{\columnwidth}{!}{
		\begin{tabular}{|c|c|c|}
		\hline
			\phantom{n}& finding a stable matching & finding an egalitarian stable matching \\ \hline
			\begin{tabular}{c} \textsc{$d$-sri} \end{tabular} & \begin{tabular}{c}in $\P$~\cite{Irv85,GI89} \end{tabular}&  \begin{tabular}{c}in $\P$ for $d=2$ ($\ast$) \\ $\NP$-hard even for $d = 3$ ($\ast$) \\ $\frac{2d+3}{7}$-approximation for $d\in \{3,4,5\}$ ($\ast$) \\ 2-approximation for $d\geq 6$ \cite{Fed92,Fed94}\end{tabular}\\  \hline
			 \begin{tabular}{c}\textsc{$d$-srti} \end{tabular} &\begin{tabular}{c} in $\P$ for $d=2$ ($\ast$)\\ $\NP$-hard even for $d = 3$ ($\ast$)\end{tabular} & \begin{tabular}{c}not well-defined (see Section~\ref{se:no_egal_srti}) \end{tabular}\\
		\hline
		\end{tabular}
		}
		\newline
\caption{Summary of results for {\sc $d$-sri} and {\sc $d$-srti}.}
\label{ta:results}
\end{table}

\paragraph{Related work.} Degree-bounded graphs, most-stable matchings and egalitarian stable matchings are widely studied concepts in the literature on matching under preferences \cite{Man13}. As already mentioned, the problem of finding a most-stable matching has been studied previously in the context of {\sc sri}~\cite{ABM06}.  In addition to the results surveyed already, the authors of~\cite{ABM06} gave an $O(m^{k+1})$ algorithm to find a matching $M$ with $|bp(M)|\leq k$ or report that no such matching exists, where $m=|E|$ and $k\geq 1$ is any integer.   Most-stable matchings have also been considered in the context of {\sc $d$-sri}~\cite{BMM12}.  The authors showed that, if $d=3$, there is some constant $c>1$ such that the problem of finding a most-stable matching is not approximable within $c$ unless $\P=\NP$.  On the other hand, they proved that the problem is solvable in polynomial time for $d\leq 2$.  The authors also gave a $(2d-3)$-approximation algorithm for the problem for fixed $d\geq 3$.  This bound was improved to $2d-4$ if the given instance satisfies an additional condition (namely the absence of a structure called an \emph{elitist odd party}).  Most-stable matchings have also been studied in the bipartite restriction of {\sc sri} called the \emph{Stable Marriage problem with Incomplete lists} ({\sc smi}) ~\cite{HIM09,BMM10}.  Since every instance of {\sc smi} admits a stable matching $M$ (and hence $bp(M)=\emptyset$), the focus in~\cite{HIM09,BMM10} was on finding maximum cardinality matchings with the minimum number of blocking edges.

Regarding the problem of finding an egalitarian stable matching in an instance of {\sc sri}, as already mentioned Feder~\cite{Fed92,Fed94} showed that this problem is $\NP$-hard, though approximable within a factor of 2.  A 2-approximation algorithm for this problem was also given independently by Gusfield and Pitt~\cite{GP92}, and by Teo and Sethuraman~\cite{TS98}.  These approximation algorithms can also be extended to the more general setting where we are given a weight function on the edges, and we seek a stable matching of minimum weight. Feder's 2-approximation algorithm requires monotone, non-negative and integral edge weights, whereas with the help of LP techniques~\cite{TS97,TS98}, the integrality constraint can be dropped, while the monotonicity constraint can be partially relaxed. Chen et al.~\cite{CHSY17} study the fixed-parameter tractability of computing egalitarian stable matchings in the setting of {\sc srti}.

\section{The Egalitarian Stable Roommates problem}
\label{se:egal}
In this section we consider the complexity and approximability of the problem of computing an egalitarian stable matching in instances of {\sc $d$-sri}.  We begin by defining the following problems.

\begin{pr}\textsc{egal $d$-sri} \\
	\inp A solvable instance $\mathcal{I} = \langle G,\mathcal O\rangle$ of {\sc $d$-sri}, where $G$ is a graph and $\mathcal O$ is a set of preference lists, each of length at most~$d$.\\
	\outp An egalitarian stable matching $M$ in~$\mathcal I$.
\end{pr}

\noindent
The decision version of {\sc egal $d$-sri} is defined as follows: 

\begin{pr}\textsc{egal $d$-sri dec} \\
	\inp $\mathcal{I} = \langle G,\mathcal O,K'\rangle$, where $\langle G,\mathcal O\rangle$ is a solvable instance $\mathcal I'$ of {\sc $d$-sri} and $K'$ is an integer.\\
	\ques Does $\mathcal I'$ admit a stable matching $M$ with $c(M)\leq K'$?
\end{pr}

\noindent In the following we give a reduction from the $\NP$-complete decision version of Minimum Vertex Cover in cubic graphs to \textsc{egal 3-sri dec}, deriving the hardness of the latter problem.

\begin{theorem}
\label{th:egal}
\textsc{egal 3-sri dec} is $\NP$-complete.
\end{theorem}

\begin{proof}
Clearly {\sc egal 3-sri dec} belongs to $\NP$.  To show $\NP$-hardness, we begin by defining the $\NP$-complete problem that we will reduce to \textsc{egal 3-sri dec}.

\begin{pr}\textsc{3-vc}\ \\
	\inp $\mathcal{I} = \langle G, K \rangle$, where $G$ is a cubic graph and $K$ is an integer.\\
	\ques Does $G$ contain a vertex cover of size at most $K$?
\end{pr}

\noindent
\textsc{3-vc} is $\NP$-complete~\cite{GJS76,MS77}. \\

\noindent
\textbf{Construction of the} \textsc{egal 3-sri dec} \textbf{instance.} 
Let $\langle G,K\rangle$ be an instance of \textsc{3-vc}, where $G=(V,E)$, $E=\{e_1,\dots,e_m\}$ and $V=\{v_1,\dots,v_n\}$. 
For each $i$ $(1\leq i\leq n)$, suppose that $v_i$ is incident to edges $e_{j_1}$, $e_{j_2}$ and $e_{j_3}$ in $G$,
where without loss of generality $j_1<j_2<j_3$.  Define $e_{i,s}=e_{j_s}$ ($1\leq s\leq 3)$.  Similarly 
for each $j$ $(1\leq j\leq m)$, suppose that $e_j=v_{i_1} v_{i_2}$, where without loss of generality $i_1<i_2$.
Define $v_{j,r}=v_{i_r}$ ($1\leq r\leq 2)$. The use of this notation is illustrated in Figure~\ref{fi:egal1}.

			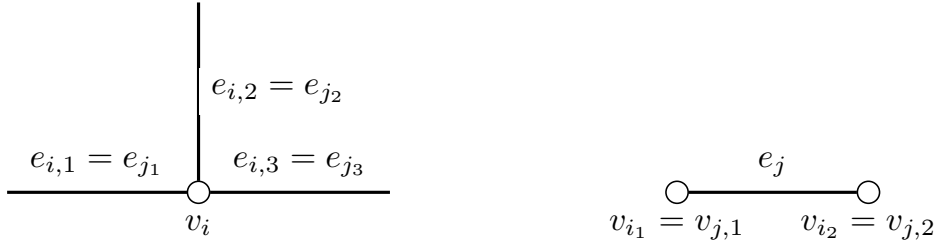
\begin{figure}[t]
			\centering
            \resizebox{0.8\columnwidth}{!}{
			\begin{tikzpicture}[scale=0.85, transform shape]
			\pgfmathsetmacro{\d}{2}

			\node[vertex, label=below:$v_i$] (vi) at (0, 0) {};
			
			\draw [thick] (vi) -- node [above=1pt, fill=white]{$e_{i,1} = e_{j_1}$}($(vi) - (\d, 0)$);
			\draw [thick] (vi)  -- node [above=1pt, fill=white]{$e_{i,3} = e_{j_3}$} ($(vi) + (\d, 0)$);
			\draw [thick] (vi)  -- node [right, fill=white]{$e_{i,2} = e_{j_2}$} ($(vi) + (0, \d)$);

			\node[vertex, label={below:$v_{i_1} = v_{j,1}$}] (vi1) at (5, 0) {};
			\node[vertex, label={below:$v_{i_2} = v_{j,2}$}] (vi2) at (7, 0) {};
			
			\draw [thick] (vi1) -- node [above=1pt, fill=white]{$e_j$}(vi2);
			\end{tikzpicture}
            }
			\caption{Notation derived from the \textsc{3-vc} instance $\langle G, K \rangle$.}
			\label{fi:egal1}
		\end{figure} 

We now construct an instance $\mathcal{I}$ of \textsc{3-sri} as follows. We define the following sets of vertices.
\begin{center}
\begin{math}
\begin{array}{lllllll}
V'&=&\{v_i^r &:& 1\leq i\leq n &\wedge& 1\leq r\leq 4\}\\
E'&=&\{e_j^s &:& 1\leq j\leq m &\wedge& 1\leq s\leq 2\}\\
W&=&\{w_i^r &:& 1\leq i\leq n &\wedge& 1\leq r\leq 4\}\\
Z&=&\{z_i^r &:& 1\leq i\leq n &\wedge& 1\leq r\leq 4\}
\end{array}
\end{math}
\end{center}

Intuitively, $v_i^r\in V'$ corresponds to vertex $v_i$ and its incident edge $e_{i,r}$, whilst $e_j^s\in E'$ corresponds to edge $e_j$ and its incident vertex~$v_{j,s}$.  The set $V'\cup E'\cup W\cup Z$ constitutes the set of agents in $\mathcal{I}$, and the preference lists of the agents are as shown in Figure~\ref{preflists}. In the preference list of an agent $v_i^r$ ($1\leq i\leq n$ and $1\leq r\leq 3$), the symbol $e(v_i^r)$ denotes the agent $e_j^s\in E'$ such that $e_j=e_{i,r}$ and $v_i=v_{j,s}$ (that is, $e_j$ is the $r$th edge incident to $v_i$ and $v_i$ is the $s$th end-vertex of $e_j$).
Similarly in the preference list of an agent $e_j^s$ ($1\leq i\leq m$ 
and $1\leq s\leq 2$), the symbol $v(e_j^s)$ denotes the agent $v_i^r\in V'$ such that $v_i=v_{j,s}$ and $e_j=e_{i,r}$ (that is, $v_i$ is the $s$th end-vertex of $e_j$ and $e_j$ is the $r$th edge incident to $v_i$). 

		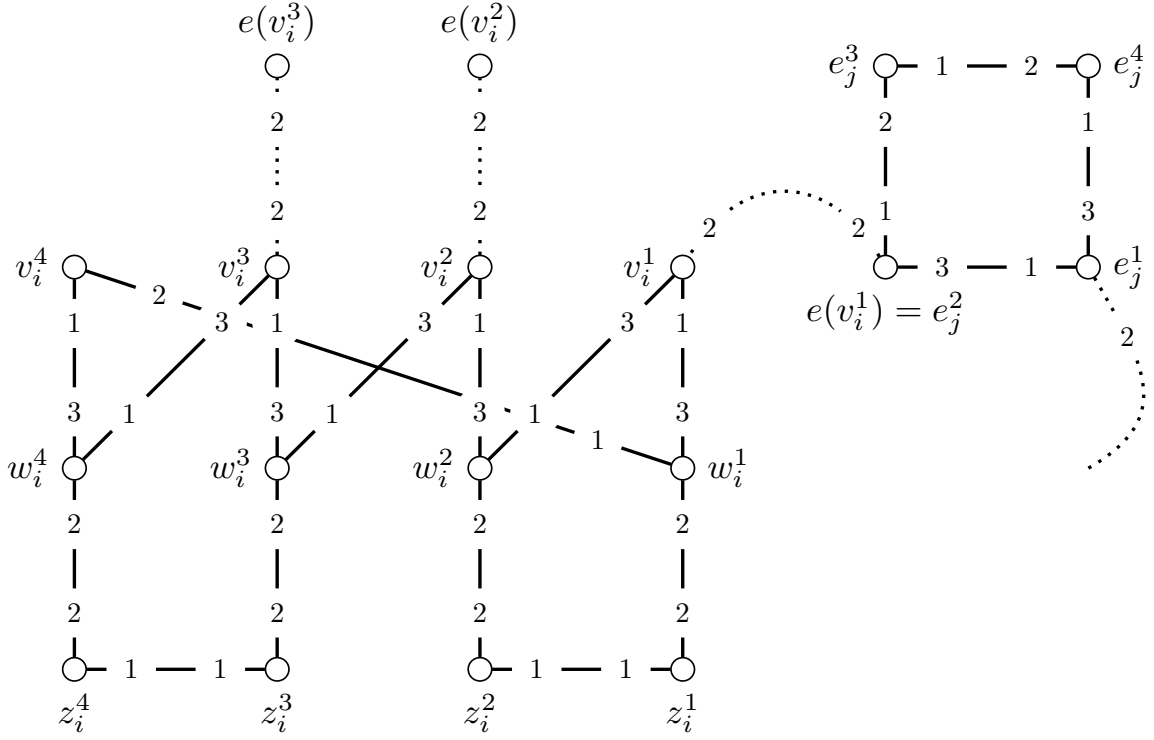
\begin{figure}[t]
			\centering
            \resizebox{\columnwidth}{!}{
			\begin{tikzpicture}[scale=0.85, transform shape]
			\pgfmathsetmacro{\d}{2}

			\node[vertex, label=below:$z_i^4$] (zi4) at (\d*3, 0) {};
			\node[vertex, label=below:$z_i^3$] (zi3) at (\d*4, 0) {};
			\node[vertex, label=below:$z_i^2$] (zi2) at (\d*5, 0) {};
			\node[vertex, label=below:$z_i^1$] (zi1) at (\d*6, 0) {};
			
			\node[vertex, label=left:$w_i^4$] (wi4) at (\d*3, \d) {};
			\node[vertex, label=left:$w_i^3$] (wi3) at (\d*4, \d) {};
			\node[vertex, label=left:$w_i^2$] (wi2) at (\d*5, \d) {};
			\node[vertex, label=right:$w_i^1$] (wi1) at (\d*6, \d) {};
			
			\node[vertex, label=left:$v_i^4$] (vi4) at (\d*3, \d*2) {};
			\node[vertex, label=left:$v_i^3$] (vi3) at (\d*4, \d*2) {};
			\node[vertex, label=left:$v_i^2$] (vi2) at (\d*5, \d*2) {};
			\node[vertex, label=left:$v_i^1$] (vi1) at (\d*6, \d*2) {};
			
			\node[vertex, label=above:$e(v_i^3)$] (evi3) at (\d*4, \d*3) {};
			\node[vertex, label=above:$e(v_i^2)$] (evi2) at (\d*5, \d*3) {};
			
			\draw [thick] (wi1) -- node[edgelabel, very near start] {1} node[edgelabel, very near end] {2} (vi4);
			
			\draw [thick] (vi4) -- node[edgelabel, near start] {1} node[edgelabel, near end] {3} (wi4);
			\draw [thick] (vi3) -- node[edgelabel, near start] {1} node[edgelabel, near end] {3} (wi3);
			\draw [thick] (vi2) -- node[edgelabel, near start] {1} node[edgelabel, near end] {3} (wi2);
			\draw [thick] (vi1) -- node[edgelabel, near start] {1} node[edgelabel, near end] {3} (wi1);
			
			\draw [thick] (zi4) -- node[edgelabel, near start] {2} node[edgelabel, near end] {2} (wi4);
			\draw [thick] (zi3) -- node[edgelabel, near start] {2} node[edgelabel, near end] {2} (wi3);
			\draw [thick] (zi2) -- node[edgelabel, near start] {2} node[edgelabel, near end] {2} (wi2);
			\draw [thick] (zi1) -- node[edgelabel, near start] {2} node[edgelabel, near end] {2} (wi1);
			
			\draw [thick] (zi1) -- node[edgelabel, near start] {1} node[edgelabel, near end] {1} (zi2);
			\draw [thick] (zi3) -- node[edgelabel, near start] {1} node[edgelabel, near end] {1} (zi4);
			
			\draw [thick, dotted] (vi2) -- node[edgelabel, near start] {2} node[edgelabel, near end] {2} (evi2);
			\draw [thick, dotted] (vi3) -- node[edgelabel, near start] {2} node[edgelabel, near end] {2} (evi3);
			\draw [thick] (vi1) -- node[edgelabel, near start] {3} node[edgelabel, near end] {1} (wi2);
			\draw [thick] (vi2) -- node[edgelabel, near start] {3} node[edgelabel, near end] {1} (wi3);
			\draw [thick] (vi3) -- node[edgelabel, near start] {3} node[edgelabel, near end] {1} (wi4);
			
			
			\node[vertex, label=left:$e_j^3$] (ej3) at (\d*7,\d*3) {};
			\node[vertex, label=right:$e_j^4$] (ej4) at (\d*8,\d*3) {};
			\node[vertex, label=right:$e_j^1$] (ej1) at (\d*8,\d*2) {};
			\node[vertex, label={below:$e(v_i^1) = e_j^2$}] (ej2) at (\d*7,\d*2) {};
			
			\draw [thick] (ej3) -- node[edgelabel, near start] {1} node[edgelabel, near end] {2} (ej4);
			\draw [thick] (ej4) -- node[edgelabel, near start] {1} node[edgelabel, near end] {3} (ej1);
			\draw [thick] (ej1) -- node[edgelabel, near start] {1} node[edgelabel, near end] {3} (ej2);
			\draw [thick] (ej2) -- node[edgelabel, near start] {1} node[edgelabel, near end] {2} (ej3);
			
			\draw [thick, dotted] (ej2) to[out=120,in=60, distance=1cm ]  node[edgelabel, very near start] {2} node[edgelabel, very near end] {2} (vi1);
						
			\draw [thick, dotted] (ej1) to[out=-60,in=30, distance=1cm ]  node[edgelabel, near start] {2} (\d*8, \d);
			\end{tikzpicture}
            }
			\caption{Part of the constructed instance of {\sc egal 3-sri dec}.}
			\label{preflists}
		\end{figure}

Finally we define some further notation in~$\mathcal{I}$. Let $K'=7m+19n+K$. The following edge sets play a particular role in our proof. Addition is taken modulo~4 here.

\begin{center}
\begin{math}
\begin{array}{lllll}
V_i^c&=&\{v_i^r w_i^r &:& 1\leq i\leq n\wedge 1\leq r\leq 4\}\\
V_i^u&=&\{v_i^r w_i^{r+1} &:&  1\leq i\leq n\wedge 1\leq r\leq 4\}\\
E_j^1&=&\{e_j^1 e_j^2, e_j^3 e_j^4 &:& 1\leq j\leq m\}\\
E_j^2&=&\{e_j^1 e_j^4, e_j^2 e_j^3 &:& 1\leq j\leq m\}\\
M_Z&=&\{z_i^1 z_i^2, z_i^3 z_i^4 &:& 1\leq i\leq n\}
\end{array}
\end{math}
\end{center}

This finishes the construction of the \textsc{egal 3-sri dec} instance~$\mathcal{I}$. In the remainder of the proof we show that $G$ has a vertex cover $C$ where $|C|\leq K$ if and only if $\mathcal{I}$ has a stable matching $M$ where $c(M)\leq K'$.

\begin{cl}
	If $G$ has a vertex cover $C$ such that $|C|=k\leq K$, then there is a stable matching $M$ in $\mathcal{I}$ such that $c(M)\leq K'$.
\end{cl}
\proof Suppose that $G$ has a vertex cover $C$ such that $|C|=k\leq K$.  We construct a matching $M$ in $\mathcal{I}$ as follows.  
For each $i$ ($1\leq i\leq n$), if $v_i\in C$, add $V_i^c$ to $M$, otherwise add $V_i^u$ to $M$. For each $j$ $(1\leq j\leq m$), if $v_{j,1}\in C$, add $E_j^2$ to $M$, otherwise add $E_j^1$ to~$M$.  Finally add the pairs in $M_Z$ to~$M$.

We now argue that $M$ is stable. Suppose that $e_j^1 e_j^4\in M$ for some $j$ ($1\leq j\leq m$).  Then $E_j^2\subseteq M$, so $v_{j,1}\in C$.  Let $v_i=v_{j,1}$.  Then by construction, $V_i^c\subseteq M$, and hence $v_i^r$ has his first choice for each $r$ ($1\leq r\leq 4$). Thus $e_j^1$ does not form a blocking edge of $M$ with $v(e_j^1)$.  The argument is similar if $e_j^1 e_j^2\in M$ for some $j$ ($1\leq j\leq m$).  Then $E_j^1\subseteq M$, so $v_{j,2}\in C$.  Let $v_i=v_{j,2}$.  Then by construction, $V_i^c\subseteq M$, and hence $v_i^r$ has his first choice for each $r$ ($1\leq r\leq 4$). Thus $e_j^2$ does not form a blocking edge of $M$ with~$v(e_j^1)$.  Now suppose that $v_i^r w_i^{r+1}\in M$ for some $i$ ($1\leq i\leq n$) and $r$ ($1\leq r\leq 3$).  Then $V_i^u\subseteq M$, so $v_i\notin C$. Let $e_j^s=e(v_i^r)$.  If $s=1$ then $v_i=v_{j,1}$.  Hence by construction of $M$, $E_j^1\subseteq M$.  Then $e_j^1$ has his first-choice partner, so $v_i^r$ does not block $M$ with~$e(v_i^r)$.  If $s=2$ then $v_i=v_{j,2}$.  As $v_{j,2}\notin C$, it follows that $v_{j,1}\in C$ as $C$ is a vertex cover.  Hence by construction of $M$, $E_j^2\subseteq M$.  Then $e_j^2$ has its first-choice partner, so $v_i^r$ does not block $M$ with~$e(v_i^r)$.  It is straightforward to verify that $M$ cannot admit any other type of blocking edge, and thus $M$ is stable in~$\mathcal{I}$. 

Clearly every agent in $\mathcal{I}$ is matched in~$M$.  We note that Theorem~4.5.2 of~\cite{GI89} implies that every stable matching in $\mathcal{I}$ matches every agent in~$\mathcal{I}$ -- we will use this fact in the next claim. We finally note that $c(M)=4k+12k+9(n-k)+2(n-k)+4(n-k)+7m+4n=7m+19n+k\leq K'$, considering the contributions from the agents matched in 
$V_i^c$, $V_i^u$ ($1\leq i\leq n$), $E_j^1$, $E_j^2$ ($1\leq j\leq m$) and $M_Z$ respectively. \qed

\begin{cl}
If there is a stable matching $M$ in $\mathcal{I}$ such that $c(M)\leq K'$ then $G$ has a vertex cover $C$ such that $|C|=k\leq K$.
\end{cl}

\proof
Suppose that $M$ is a stable matching in $\mathcal{I}$ such that $c(M)\leq K'$.  We construct a set of vertices $C$ in $G$ as follows. As $M$ matches every agent in $\mathcal{I}$, then for each $i$ $(1\leq i\leq n$), either $V_i^c\subseteq M$ or $V_i^u\subseteq M$.  In the former case add $v_i$ to~$C$. Also, for each $j$ ($1\leq j\leq m$), as $M$ matches every agent in $\mathcal{I}$, either
$E_j^1\subseteq M$ or $E_j^2\subseteq M$.  Finally, it follows that $M_Z\subseteq M$.  

We now argue that $C$ is a vertex cover. Let $j$ ($1\leq j\leq m$) be given and suppose that $v_{j,1}\notin C$ and $v_{j,2}\notin C$.  Suppose firstly that $E_j^1\subseteq M$. Let $v_i=v_{j,2}$. Then $V_i^u\subseteq M$ by construction of $C$, so that $e_j^2$ blocks $M$ with $v(e_j^2)$, a contradiction.  Now suppose that $E_j^2\subseteq M$. Let $v_i=v_{j,1}$. Then $V_i^u\subseteq M$ by construction of $C$, so that $e_j^1$ blocks $M$ with $v(e_j^1)$, a contradiction.  Hence $C$ is a vertex cover in~$G$.

Moreover if $k=|C|$ then given the composition of $M$, as noted in the previous claim, $c(M)=7m+19n+k$, and since $c(M)\leq K'$ it follows that $k\leq K$. 
\end{proof}

\noindent  Theorem \ref{th:egal} immediately implies the following result.

\begin{cor}
\label{co:egal}
\textsc{egal 3-sri} is $\NP$-hard.
\end{cor}

We remark that \textsc{egal 2-sri} is trivially solvable in polynomial time: the components of the graph are paths and cycles in this case, and the cost of a stable matching selected in one component is not affected by the matching edges chosen in another component.  Therefore we can deal with each path and cycle separately, minimising the cost of a stable matching in each.  Paths and odd cycles admit exactly one stable matching (recall that (i) the instance is assumed to be solvable, and (ii) the set of matched agents is the same in all stable matchings \cite[Theorem 4.5.2]{GI89}), whilst even cycles admit at most two stable matchings (to find them, test each of the two perfect matchings for stability) -- we can just pick the stable matching with lower cost in such a case.  The following result is therefore immediate.

\begin{prop}
\textsc{egal 2-sri} admits a linear-time algorithm.
\end{prop}

Corollary~\ref{co:egal} naturally leads to the question of the approximabilty of \textsc{egal $d$-sri}. As mentioned in the Introduction, Feder~\cite{Fed92,Fed94} provided a 2-approximation algorithm for the problem of finding an egalitarian stable matching in an instance of {\sc sri}. As Theorems~\ref{th:egal3}, \ref{th:egal4} and \ref{th:egal5} show, this bound can be improved for instances with bounded-length preference lists.

\begin{theorem}
	\label{th:egal3}
	\textsc{egal 3-sri} is approximable within $9/7$.
\end{theorem}

\begin{proof}
Let $\mathcal I$ be an instance of {\sc 3-sri} and let $\megal$ denote an egalitarian stable matching in~$\mathcal I$.  First we show that any stable matching in $\mathcal I$ is a $4/3$-approximation to~$\megal$. We then focus on the worst-case scenario when this ratio $4/3$ is in fact realised. Then we design a weight function on the edges of the graph and apply Teo and Sethuraman's 2-approximation algorithm~\cite{TS97,TS98} to find an approximate solution $M'$ to a minimum weight stable matching $M_{opt}$ for this weight function.  This weight function helps $M'$ to avoid the worst case for the $4/3$-approximation for a significant amount of the matching edges.  We will ultimately show that $M'$ is in fact a $9/7$-approximation to~$\megal$.

\begin{cl}
\label{cl:43}
In an instance of \textsc{egal 3-sri}, any stable matching approximates $c(\megal)$ within a factor of~$4/3$.
\end{cl}

\begin{proof}
Let $M$ be an arbitrary stable matching in~$\mathcal I$. Call an edge $uv$ an \emph{$(i,j)$-pair} $(i\leq j)$ if $v$ is $u$'s $i$th choice and $u$ is $v$'s $j$th choice. By Theorem~4.5.2 of~\cite{GI89}, the set of agents matched in $\megal$ is identical to the set of agents matched in~$M$. We will now study the worst approximation ratios in all cases of $(i,j)$-pairs, given that $1 \leq i \leq j \leq 3$ in \textsc{3-sri}.

\begin{itemize}

\item If $uv\in \megal$ is a $(1,1)$-pair then $u$ and $v$ contribute 2 to $c(\megal)$ and also 2 to $c(M)$ since they must be also be matched in $M$ (and in every stable matching).

\item If $uv\in \megal$ is a $(1,2)$-pair then $u$ and $v$ contribute 3 to $c(\megal)$ and at most 4 to~$c(M)$. Since, if $uv \notin M$, then $v$ must be matched to his 1st choice and $u$ to his 2nd or 3rd, because one of $u$ and $v$ must be better off and the other must be worse off in $M$ than in~$\megal$.

\item If $uv\in \megal$ is a $(1,3)$-pair then $u$ and $v$ contribute 4 to $c(\megal)$ and at most 5 to~$c(M)$. Since, if $uv \notin M$, then $v$ must be matched to his 1st or 2nd choice and $u$ to his 2nd or 3rd.

\item If $uv\in \megal$ is a $(2,2)$-pair then $u$ and $v$ contribute 4 to $c(\megal)$ and at most 4 to~$c(M)$. Since, if $uv \notin M$, then one must be matched to his 1st choice and the other to his 3rd.

\item If $uv\in \megal$ is a $(2,3)$-pair then $u$ and $v$ contribute 5 to $c(\megal)$ and at most 5 to~$c(M)$. Since, if $uv \notin M$, then $v$ must be matched to his 1st or 2nd choice and $u$ to his 3rd.

\item If $uv\in \megal$ is a (3,3)-pair then $u$ and $v$ contribute 6 to $c(\megal)$ and also 6 to $c(M)$ since they must be also be matched in $M$ (and in every stable matching -- this follows by \cite[Lemma 4.3.9]{GI89}).
\end{itemize}

It follows that, for every pair $uv\in \megal$,
\begin{eqnarray*}
\frac{\rank(u,M(u))+\rank(v,M(v))}{\rank(u,\megal(u))+\rank(v,\megal(v))} & = & \frac{\rank(u,M(u))+\rank(v,M(v))}{\rank(u,v)+\rank(v,u)}\\ 
& \leq & 4/3.
\end{eqnarray*}

Hence $c(M) / c(\megal) \leq 4/3$ and Claim~\ref{cl:43} is proved.
\end{proof}

As shown in Claim~\ref{cl:43}, the only case when the approximation ratio $4/3$ is reached is where $\megal$ consists of (1,2)-pairs exclusively, while the stable matching output by the approximation algorithm contains (1,3)-pairs only. We will now present an algorithm that either delivers a stable solution $M'$ containing at least a significant amount of the (1,2)-pairs in $\megal$ or a certificate that $\megal$ contains only a few (1,2)-pairs and thus any stable solution is a good approximation.

	To simplify our proof, we execute some basic pre-processing of the input graph. If there are any (1,1)-pairs in $G$, then these can be fixed, because they occur in every stable matching and thus can only lower the approximation ratio. Similarly, if an arbitrary stable matching contains a (3,3)-pair, then this edge appears in all stable matchings and thus we can fix it. Those (3,3)-pairs that do not belong to the set of stable edges can be deleted from the graph. From this point on, we assume that no edge is ranked first or last by both of its end vertices in $G$ and prove the approximation ratio for such graphs.

Take the following weight function on all $uv \in E$:
	\[   
w(uv) = 
     \begin{cases}
       0 &\quad\text{if } uv\text{ is a (1,2)-pair},\\ 
       1 &\quad\text{otherwise.}
       \end{cases}
\]
	
	We designed $w(uv)$ to fit the necessary U-shaped condition of Teo and Sethuraman's 2-approximation algorithm~\cite{TS97,TS98}. This condition on the weight function is as follows. We are given a function $f_p$ on the neighbouring edges of a vertex~$p$. Function $f_p$ is \emph{U-shaped} if it is non-negative and there is a neighbour $q$ of $p$ so that $f_p$ is monotone decreasing on neighbours in order of $p$'s preference until $q$, and $f_p$ is monotone increasing on neighbours in order of $p$'s preference after~$q$. The approximation guarantee of Teo and Sethuraman's algorithm holds for an edge weight function $w(uv)$ if for every edge $uv \in E$, $w(uv)$ can be written as $w(uv) = f_u(uv) + f_v(uv)$, where $f_u$ and $f_v$ are U-shaped functions.
	
	Our $w(uv)$ function is clearly U-shaped, because at each vertex the sequence of edges in order of preference is either monotone increasing or it is $(1,0,1)$. Since $w$ itself is U-shaped, it is easy to decompose it into a sum of U-shaped $f_v$ functions, for example by setting $f_v(uv) = f_u(uv) = \frac{w(uv)}{2}$ for every edge~$uv$. 
	
	Let $M$ denote an arbitrary stable matching, let $M^{(1,2)}$ be the set of (1,2)-pairs in $M$, and let $\mopt$ be a minimum weight stable matching with respect to the weight function~$w(uv)$. Since $\mopt$ is by definition the stable matching with the largest number of (1,2)-pairs, $|\mopt^{(1,2)}| \geq |\megal^{(1,2)}|$. We also know that $w(M) = |M| - |M^{(1,2)}|$ for every stable matching~$M$.
	
	Due to Teo and Sethuraman's approximation algorithm~\cite{TS97,TS98}, it is possible to find a stable matching $M'$ whose weight approximates $w(\mopt)$ within a factor of~2. Formally,    
    $$|M| - |M'^{(1,2)}| = w(M') \leq 2w(\mopt) = 2|M| - 2|\mopt^{(1,2)}|.$$
    This gives us a lower bound on $|M'^{(1,2)}|$.    
    \begin{equation}
    \label{eq:M'12}
    |M'^{(1,2)}| \geq 2|\mopt^{(1,2)}| - |M| \geq 2|\megal^{(1,2)}| - |M|
    \end{equation}
    
    We distinguish two cases from here on, depending on the sign of the term on the right. In both cases, we establish a lower bound on $c(\megal)$ and an upper bound on~$c(M')$. These will give the desired upper bound of 9/7 on $\frac{c(M')}{c(\megal)}$.  
    \begin{enumerate}[1)]
    
    \item $2|\megal^{(1,2)}| - |M| \leq 0$
    
    The derived lower bound for $|M'^{(1,2)}|$ is negative or zero in this case. Yet we know that at most half of the edges in $\megal$ are (1,2)-pairs, and $c(e) \geq 4$ for the rest of the edges in~$\megal$. Let us denote $|M| -2 |\megal^{(1,2)}| \geq 0$ by~$x$. Thus, $|\megal^{(1,2)}|=\frac{|M|-x}{2}$. 
		\begin{equation}
        \label{eq:1Megal}
        c(\megal) \geq \frac{|M|-x}{2} \cdot 3 + \frac{|M|+x}{2} \cdot 4 = 3.5 |M| +0.5x
		\end{equation}
        
		We use our arguments in the proof of Claim~\ref{cl:43} to derive that an arbitrary stable matching approximates $c(\megal)$ on the $\frac{|M|-x}{2}$ (1,2)-edges within a ratio of $\frac{4}{3}$, while its cost on the remaining $\frac{|M|+x}{2}$ edges is at most~5. These imply the following inequalities for an arbitrary stable matching~$M$. 
    \begin{equation}
        \label{eq:1M}
        c(M) \leq \frac{|M|-x}{2} \cdot 3 \cdot \frac{4}{3} + \frac{|M|+x}{2} \cdot 5 = 4.5 |M| + 0.5x
        \end{equation}
		
		We now combine (\ref{eq:1Megal}) and (\ref{eq:1M}). The last inequality holds for all $x  \geq 0$.
\begin{equation*}
\begin{split}
\frac{c(M)}{c(\megal)} &\leq
\frac{4.5|M| +0.5x}{3.5|M| +0.5x} \leq \frac{9}{7}\
\end{split}
\end{equation*}
    
    \item $2|\megal^{(1,2)}| - |M| >0$
    
    Let us denote $2|\megal^{(1,2)}| - |M|$ by~$\hat{x}$. Notice that $|\megal^{(1,2)}| = \frac{\hat{x} + |M|}{2}$. We can now express the number of edges with cost~3, and at least 4 in~$\megal$. 
\begin{eqnarray}    
c(\megal) & \geq &
3  \cdot \frac{\hat{x} + |M|}{2} + 4  \cdot \left(|M| - \frac{\hat{x} + |M|}{2}\right) \nonumber \\
& = & 3.5|M| -0.5\hat{x}
\label{eq:2Megal}
\end{eqnarray}
Let $|M'^{(1,2)}|=z_1$.  Then exactly $z_1$ edges in $M'$ have cost~3.  It follows from (\ref{eq:M'12}) that $z_1\geq \hat{x}$.  Suppose that $z_2\leq z_1$ edges in $M'^{(1,2)}$ correspond to edges in $\megal^{(1,2)}$.  Recall that $|\megal^{(1,2)}|=\frac{\hat{x}+|M|}{2}$.  The remaining $\frac{|M|+\hat{x}}{2}-z_2$ edges in $\megal^{(1,2)}$ have cost at most 4 in~$M'$.  This leaves $|M|-|\megal^{(1,2)}|-(z_1-z_2)=\frac{|M|-\hat{x}}{2}-z_1+z_2$ edges in $\megal$ that are as yet unaccounted for; these have cost at most 5 in both $\megal$ and~$M'$.  We thus obtain:

\begin{eqnarray}
c(M') & \leq & 
3z_1 + 4 \left(\frac{|M| + \hat{x}}{2}-z_2\right) + 5 \left(\frac{|M| - \hat{x}}{2}-z_1+z_2\right) \nonumber \\
&=& 4.5|M| - 0.5\hat{x} -2 z_1 + z_2 \nonumber\\
&\leq& 4.5|M| -1.5\hat{x}
\label{eq:2M}
\end{eqnarray}

     Combining (\ref{eq:2Megal}) and (\ref{eq:2M}) delivers the following bound.
\begin{equation*}
\begin{split}
\frac{c(M')}{c(\megal)} &\leq
\frac{4.5|M| -1.5\hat{x}}{3.5|M| -0.5\hat{x}} < \frac{9}{7}\
\end{split}
\end{equation*}
    The last inequality holds for every $\hat{x} > 0$.
    \end{enumerate}
    
We derived that $M'$, the 2-approximate solution with respect to the weight function $w(uv)$ delivers a $\frac{9}{7}$-approximation in both cases.
 \end{proof}

Using analogous techniques we can establish similar approximation bounds for \textsc{egal 4-sri} and \textsc{egal 5-sri}, as follows.

\begin{theorem}
\label{th:egal4}
	\textsc{egal 4-sri} is approximable within $11/7$.
\end{theorem}

\begin{proof}
We start with a statement analogous to Claim~\ref{cl:43}.
\begin{cl}
In an instance of \textsc{egal 4-sri}, any stable matching approximates $c(\megal)$ within a factor of~$5/3$. 
\end{cl}
\proof	As earlier, we can fix all (1,1)-pairs and eliminate all (4,4)-pairs from the instance. Table~\ref{ta:egal4} contains all cases for $uv$ edges in $\megal$ and the corresponding costs in an arbitrary stable matching.\qed
	\begin{table}[t!]
	\setlength{\tabcolsep}{6pt}
		\centering
			\begin{tabular}{cccc}
					$uv$ & worst case cost at $u$ & worst case cost at $v$ & cost ratio\\ 
					\hline
					(1,2) & 4 & 1 & 5/3\\
					(1,3) & 4 & 2 & 6/4\\
					(1,4) & 4 & 3 & 7/5\\
					(2,2) & 4 & 1 & 5/4\\
					(2,3) & 4 & 2 & 6/5\\
					(2,4) & 4 & 3 & 7/6\\
					(3,3) & 4 & 2 & 6/6\\
					(3,4) & 4 & 3 & 7/7\\
                    \hline
        	\end{tabular}
		\caption{$uv$ edges and the corresponding costs in \textsc{egal 4-sri}.}
	\label{ta:egal4}
	\end{table}

We define the same weight function $w(uv)$ as in the proof of Theorem~\ref{th:egal3}. We remark here that $w(uv)$ remains U-shaped for preference lists of length 4, because at each vertex the sequence of edges in order of preference is either monotone increasing or it is (1,0,1,1). Since we derived Inequality~(\ref{eq:M'12}) without using the bounded degree property, it holds for \textsc{egal 4-sri} as well. We distinguish two cases based on the sign of $2|\megal^{(1,2)}| - |M|$.

\begin{enumerate}[1)]
\item $2|\megal^{(1,2)}| - |M| \leq 0$\\
 Let us denote $|M| -2 |\megal^{(1,2)}| \geq 0$ by~$x$. Thus, $|\megal^{(1,2)}|=\frac{|M|-x}{2}$. Furthermore, let $y$ denote the number of edges with cost at least 5 in~$\megal$. 
	\begin{eqnarray*}
        c(\megal) & \geq & \frac{|M|-x}{2} \cdot 3 + \left( \frac{|M|+x}{2} -y \right)\cdot 4 +5y\\ 
        & = & 3.5 |M| +0.5x +y
	\end{eqnarray*}
                
    \begin{equation*}
        c(M) \leq \frac{|M|-x}{2} \cdot 3 \cdot \frac{5}{3} + \left( \frac{|M|+x}{2} -y \right) \cdot 6 + 7y= 5.5 |M| + 0.5x +y
    \end{equation*}
    
\begin{equation*}
\begin{split}
\frac{c(M)}{c(\megal)} &\leq
\frac{5.5|M| +0.5x+y}{3.5|M| +0.5x+y} \leq \frac{11}{7}\
\end{split}
\end{equation*}

\item $2|\megal^{(1,2)}| - |M| > 0$\\
Let $\hat{x}$ denote $2|\megal^{(1,2)}| - |M|$ and $y$ the number of edges with cost at least 5 in~$\megal$. Due to Inequality~(\ref{eq:M'12}), we know that at least $\hat{x}$ $(1,2)$-pairs in $\megal$ correspond to edges of cost~3 in~$M'$. The remaining $\frac{|M|-\hat{x}}{2}$ $(1,2)$-pairs in $\megal$ correspond to edges of cost at most 5 in~$M'$.
 
$$c(\megal) \geq \frac{\hat{x}+|M|}{2} \cdot 3 +4 \cdot(\frac{|M|-\hat{x}}{2}-y) +5y = 3.5 |M| -0.5 \hat{x} +y$$

$$c(M') \leq 3\hat{x} + 5\cdot \frac{|M|-\hat{x}}{2} +6\cdot (\frac{|M|-\hat{x}}{2} -y) + 7y = 5.5|M|-2.5\hat{x} +y$$
\end{enumerate}

\begin{equation*}
\begin{split}
\frac{c(M')}{c(\megal)} &\leq
\frac{5.5|M| -2.5\hat{x} +y}{3.5|M| -0.5\hat{x} +y} < \frac{11}{7}\
\end{split}
\end{equation*}
~ \vspace{-8mm} \\
\end{proof}

\begin{theorem}
\label{th:egal5}
	\textsc{egal 5-sri} is approximable within $13/7$.
\end{theorem}

\begin{proof}
Again we start with a statement analogous to Claim~\ref{cl:43}.
\begin{cl}
In an instance of \textsc{egal 5-sri}, any stable matching approximates $c(\megal)$ within a factor of~$2$.  
\end{cl}
\proof
	As earlier, we can fix all (1,1)-pairs and eliminate all (5,5)-pairs from the instance. Table~\ref{ta:egal5} contains all cases for $uv$ edges in $\megal$ and the corresponding costs in an arbitrary stable matching. \qed
	
	We remark that $w(uv)$ remains U-shaped for preference lists of length~5, because at each vertex the sequence of edges in order of preference is either monotone increasing or it is (1,0,1,1,1). We observe that Inequality~(\ref{eq:M'12}) holds for \textsc{egal 5-sri} as well. Thus we distinguish two cases based on the sign of $2|\megal^{(1,2)}| - |M|$. 
	
	\begin{table}[htbp]
	\setlength{\tabcolsep}{6pt}
		\centering
			\begin{tabular}{cccc}
					$uv$ & worst case cost at $u$ & worst case cost at $v$ & cost ratio\\ 
					\hline
					(1,2) & 5 & 1 & 6/3\\
					(1,3) & 5 & 2 & 7/4\\
					(1,4) & 5 & 3 & 8/5\\
					(1,5) & 5 & 4 & 9/6\\
					(2,2) & 5 & 1 & 6/4\\
					(2,3) & 5 & 2 & 7/5\\
					(2,4) & 5 & 3 & 8/6\\
					(2,5) & 5 & 4 & 9/7\\
					(3,3) & 5 & 2 & 7/6\\
					(3,4) & 5 & 3 & 8/7\\
					(3,5) & 5 & 4 & 9/8\\
					(4,4) & 5 & 3 & 8/8\\
					(4,5) & 5 & 4 & 9/9\\
                    \hline
			\end{tabular}
		\caption{$uv$ edges and the corresponding costs in \textsc{egal 5-sri}.}
		\label{ta:egal5}
	\end{table}

\begin{enumerate}[1)]
\item $2|\megal^{(1,2)}| - |M| \leq 0$\\
 Let us denote $|M| -2 |\megal^{(1,2)}| \geq 0$ by~$x$. Thus, $|\megal^{(1,2)}|=\frac{|M|-x}{2}$. Furthermore, let $y$ be the number of edges with cost~5 and $z$ the number of edges with cost at least 6 in~$\megal$.
	\begin{eqnarray*}
        c(\megal) & \geq & \frac{|M|-x}{2} \cdot 3 + \left( \frac{|M|+x}{2} -y -z \right)\cdot 4 +5y + 6z \\ 
        & = & 3.5 |M| +0.5x +y + 2z
	\end{eqnarray*}
                
    \begin{eqnarray*}
        c(M) & \leq & \frac{|M|-x}{2} \cdot 3 \cdot \frac{6}{3} + \left( \frac{|M|+x}{2} -y -z \right) \cdot 7 + 8y +9z \\ 
        & = & 6.5 |M| + 0.5x +y + 2z
    \end{eqnarray*}
    
\begin{equation*}
\begin{split}
\frac{c(M)}{c(\megal)} &\leq
\frac{6.5 |M| + 0.5x +y + 2z}{3.5 |M| +0.5x +y + 2z} \leq \frac{13}{7}\
\end{split}
\end{equation*}

\item $2|\megal^{(1,2)}| - |M| > 0$\\
Let $\hat{x}$ denote $2|\megal^{(1,2)}| - |M|$, $y$ the number of edges with cost~5 and $z$ the number of edges with cost at least 6 in~$\megal$.
 
$$c(\megal) \geq \frac{\hat{x}+|M|}{2} \cdot 3 +4 \cdot(\frac{|M|-\hat{x}}{2}-y-z) +5y +6z= 3.5 |M| -0.5 \hat{x} +y +2z$$

$$c(M') \leq 3\hat{x} + 6\cdot \frac{|M|-\hat{x}}{2} +7\cdot (\frac{|M|-\hat{x}}{2} -y-z) + 8y +9z = 6.5|M|-3.5\hat{x} +y +2z$$
\end{enumerate}

\begin{equation*}
\begin{split}
\frac{c(M')}{c(\megal)} &\leq
\frac{6.5|M| -3.5\hat{x} +y +2z}{3.5|M| -0.5\hat{x} +y +2z} < \frac{13}{7}\
\end{split}
\end{equation*}
~ \vspace{-8mm} \\
\end{proof}

Using a similar reasoning for each $d\geq 6$, our approach gives a $c_d$-approxi\-mation algorithm for {\sc egal $d$-sri} where $c_d>2$.  In these cases the 2-approxi\-mation algorithm of Feder \cite{Fed92,Fed94} should be used instead.

\section{Solvability and most-stable matchings in {\sc $d$-srti}}
\label{se:srt}
In this section we study the complexity and approximability of the problem of deciding whether an instance of {\sc $d$-srti} admits a stable matching, and the problem of finding a most-stable matching given an instance of {\sc $d$-srti}.

We begin by defining two problems that we will be studying in this section from the point of view of complexity and approximability. 

\begin{pr}\textsc{solvable $d$-srti}\ \\
	\inp $\mathcal{I} = \langle G, \mathcal O\rangle$, where $G$ is a graph and $\mathcal O$ is a set of preference lists, each of length at most $d$, possibly involving ties.\\
	\ques Is $\mathcal{I}$ solvable?
\end{pr}

\begin{pr} \textsc{min bp $d$-srti}\ \\
	\inp An instance $\mathcal I$ of {\sc $d$-srti}.\\
	\outp A matching $M$ in $\mathcal{I}$ such that $|bp(M)|=bp(\mathcal I)$.
\end{pr}

\noindent
We will show that \textsc{solvable 3-srti} is $\NP$-complete and \textsc{min bp 3-srti} is hard to approximate.  In both cases we will use a reduction from the following satisfiability problem:

\begin{pr}\textsc{(2,2)-e3-sat}\ \\
	\inp $\mathcal{I} = B$, where $B$ is a Boolean formula in CNF, in which each clause comprises exactly 3 literals and each variable appears exactly twice in unnegated and exactly twice in negated form. \\
	\ques Is there a truth assignment satisfying~$B$?
    	\label{pr:22e3sat}
\end{pr}

\noindent \textsc{(2,2)-e3-sat}\ is $\NP$-complete, as shown by Berman et al.~\cite{BKS03}.  We begin with the hardness of \textsc{solvable 3-srti}.

\begin{theorem}
\label{th:3-srti}
	\textsc{solvable 3-srti} is $\NP$-complete.
\end{theorem}

\begin{proof}
Clearly \textsc{solvable 3-srti} belongs to~$\NP$.  To show $\NP$-hardness, we reduce from \textsc{(2,2)-e3-sat} as defined in Problem \ref{pr:22e3sat}.  Let $B$ be a given instance of \textsc{(2,2)-e3-sat}, where $X = \left\{x_1, x_2, \dots , x_n\right\}$ is the set of variables and $C = \left\{c_1, c_2, \dots, c_m\right\}$ is the set of clauses.  We form an instance $\mathcal{I} = (G, \mathcal O)$ of \textsc{3-srti} as follows. Graph $G$ consists of a \emph{variable gadget} for each $x_i$ $(1 \leq i \leq n)$, a \emph{clause gadget} for each $c_j$ $(1 \leq j \leq m)$ and a set of \emph{interconnecting edges} between them; these different parts of the construction, together with the preference orderings that constitute $\mathcal O$, are shown in Figure~\ref{fi:srtigadgets} and will be described in more detail below. 
	
		\begin{figure}[t]
			\centering
            \hspace{-1mm}
            \resizebox{1\columnwidth}{!}{
			\begin{tikzpicture}[scale=0.85, transform shape]
			\pgfmathsetmacro{\d}{2}
			
			\node[vertex, label=above:$y_j^3$] (y3) at (\d, \d) {};
			\node[vertex, label=above:$y_j^4$] (y4) at (2*\d, \d) {};
			\node[vertex, label=below:$y_j^2$] (y2) at (\d, 0) {};
			\node[vertex, label=below:$y_j^1$] (y1) at (\d*2, 0) {};
			\node[vertex, label=below:$p_j^3$] (p3) at (\d*3, 0) {};
			\node[vertex, label=below:$b_j^3$] (b3) at (\d*4, 0) {};
			\node[vertex, label=below:$a_j^3$] (a3) at (\d*5, 0) {};
			\node[vertex, label=below:$q_j^3$] (q3) at (\d*6, 0) {};
			\node[vertex, label=below:$z_j^1$] (x1) at (\d*7, 0) {};
			\node[vertex, label=below:$z_j^2$] (x2) at (\d*8, 0) {};
			\node[vertex, label=above:$z_j^3$] (x3) at (\d*8, \d) {};
			\node[vertex, label=above:$z_j^4$] (x4) at (\d*7, \d) {};
			
			\node[vertex, label=left:$p_j^2$] (p2) at (\d*3, \d) {};
			\node[vertex, label=below:$b_j^2$] (b2) at (\d*4, \d) {};
			\node[vertex, label=below:$a_j^2$] (a2) at (\d*5, \d) {};
			\node[vertex, label=right:$q_j^2$] (q2) at (\d*6, \d) {};
			
			\node[vertex, label=above:$p_j^1$] (p1) at (\d*3, \d*2) {};
			\node[vertex, label=above:$b_j^1$] (b1) at (\d*4, \d*2) {};
			\node[vertex, label=above:$a_j^1$] (a1) at (\d*5, \d*2) {};
			\node[vertex, label=above:$q_j^1$] (q1) at (\d*6, \d*2) {};
			
			\draw [thick] (y3) -- node[edgelabel, near start] {1} node[edgelabel, near end] {2} (y4);
			\draw [thick] (y3) -- node[edgelabel, near start] {2} node[edgelabel, near end] {2} (y2);
			\draw [thick] (y2) -- node[edgelabel, near start] {3} node[edgelabel, near end] {1} (y4);
			
			\draw [thick] (x3) -- node[edgelabel, near start] {1} node[edgelabel, near end] {2} (x4);
			\draw [thick] (x3) -- node[edgelabel, near start] {2} node[edgelabel, near end] {2} (x2);
			\draw [thick] (x2) -- node[edgelabel, near start] {3} node[edgelabel, near end] {1} (x4);
			
			\draw [thick] (y2) -- node[edgelabel, near start] {1} node[edgelabel, near end] {2} (y1);
			\draw [thick] (y1) -- node[edgelabel, near start] {1} node[edgelabel, near end] {3} (p3);
			
			\draw [thick] (x2) -- node[edgelabel, near start] {1} node[edgelabel, near end] {2} (x1);
			\draw [thick] (x1) -- node[edgelabel, near start] {1} node[edgelabel, near end] {3} (q3);
			
			\draw [thick] (p3) -- node[edgelabel, near start] {2} node[edgelabel, near end] {1} (b3);
			\draw [thick] (b3) -- node[edgelabel, near start] {1} node[edgelabel, near end] {1} (a3);
			\draw [thick] (a3) -- node[edgelabel, near start] {3} node[edgelabel, near end] {2} (q3);
			
			\draw [thick] (p1) -- node[edgelabel, near start] {1} node[edgelabel, near end] {1} (b1);
			\draw [thick] (b1) -- node[edgelabel, near start] {1} node[edgelabel, near end] {1} (a1);
			\draw [thick] (a1) -- node[edgelabel, near start] {3} node[edgelabel, near end] {1} (q1);
			
			\draw [thick] (p1) -- node[edgelabel, near start] {2} node[edgelabel, near end] {1} (b2);
			\draw [thick] (b2) -- node[edgelabel, near start] {1} node[edgelabel, near end] {1} (a2);
			\draw [thick] (a2) -- node[edgelabel, near start] {3} node[edgelabel, near end] {2} (q1);
			
			\draw [thick] (p1) -- node[edgelabel, near start] {3} node[edgelabel, near end] {1} (p2);
			\draw [thick] (p2) -- node[edgelabel, near start] {2} node[edgelabel, near end] {1} (p3);
			
			\draw [thick] (q1) -- node[edgelabel, near start] {3} node[edgelabel, near end] {1} (q2);
			\draw [thick] (q2) -- node[edgelabel, near start] {2} node[edgelabel, near end] {1} (q3);
			
			\node[vertex, label=left:$v_i^1$] (v1) at (\d*7,\d*3) {};
			\node[vertex, label=right:$v_i^2$] (v2) at (\d*8,\d*3) {};
			\node[vertex, label=right:$v_i^3$] (v3) at (\d*8,\d*2) {};
			\node[vertex, label=below:$v_i^4$] (v4) at (\d*7,\d*2) {};
			
			\draw [thick] (v1) -- node[edgelabel, near start] {1} node[edgelabel, near end] {3} (v2);
			\draw [thick] (v2) -- node[edgelabel, near start] {1} node[edgelabel, near end] {3} (v3);
			\draw [thick] (v3) -- node[edgelabel, near start] {1} node[edgelabel, near end] {3} (v4);
			\draw [thick] (v4) -- node[edgelabel, near start] {1} node[edgelabel, near end] {3} (v1);
			
			\draw [thick, dotted] (v4) to[out=120,in=60, distance=2cm ]  node[edgelabel, very near start] {2} node[edgelabel, very near end] {2} (a1);
			
			\draw [thick, dotted] (v2) to[out=120,in=-20, distance=0.8cm ]  node[edgelabel, near start] {2} (\d*7.2, \d*3.3);
			
			\draw [thick, dotted] (v3) to[out=-60,in=20, distance=0.8cm ]  node[edgelabel, near start] {2} (\d*7.2, \d*1.7);
			
			\draw [thick, dotted] (v1) to[out=120,in=-20, distance=0.8cm ]  node[edgelabel, near start] {2} (\d*6.2, \d*3.3);
			
			\draw [thick, dotted] (a2) to[out=20,in=-160, distance=0.8cm ]  node[edgelabel, near start] {2} (\d*5.7, \d*1.2);
			
			\draw [thick, dotted] (a3) to[out=50,in=-160, distance=0.8cm ]  node[edgelabel, near start] {2} (\d*5.7, \d*0.5);
			\end{tikzpicture}
            }
			\caption{Clause and variable gadgets for \textsc{3-srti}. The dotted edges are the interconnecting edges. The notation used for edge $a_j^1v_i^4$ implies that the first literal of the corresponding clause $c_j$ is the second occurrence of the corresponding variable $x_i$ in negated form.}
			\label{fi:srtigadgets}
		\end{figure}
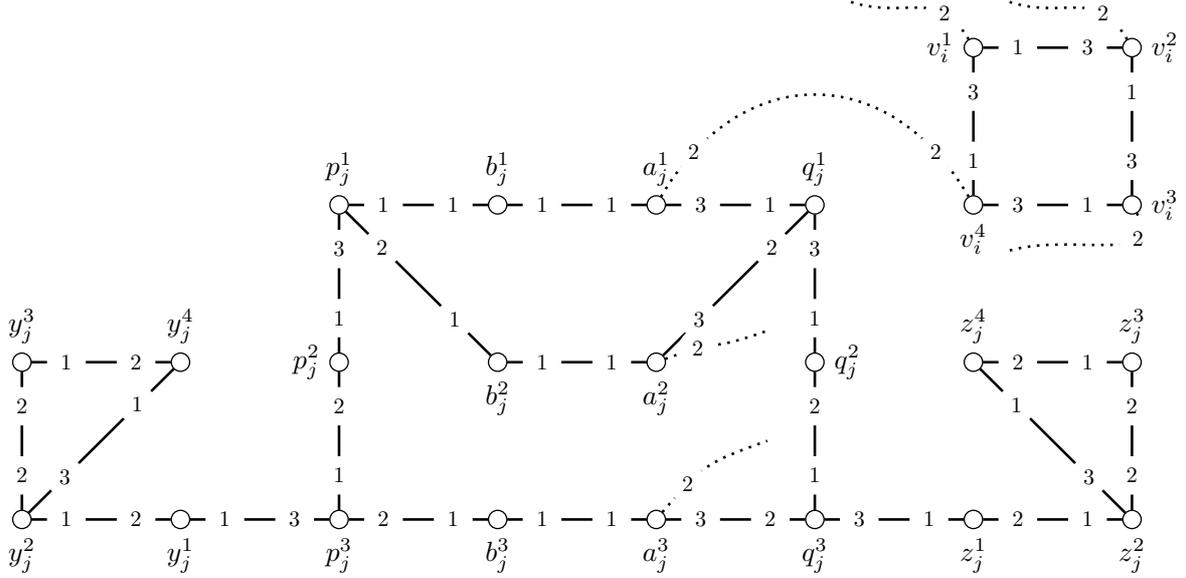
	
	When constructing $G$, we will keep track of the order of the three literals in each clause of $B$ and the order of the two unnegated and two negated occurrences of each variable in~$B$. Each of these four occurrences of each variable is represented by an interconnecting edge.
	
	A variable gadget for a variable $x_i$ ($1\leq i\leq n$) of $B$ comprises the 4-cycle $\langle v_i^1, v_i^2, v_i^3, v_i^4 \rangle$ with cyclic preferences. Each of these four vertices is incident to an interconnecting edge. These edges end at specific vertices of clause gadgets.  The clause gadget for a clause $c_j$ ($1\leq j\leq m$) contains 20 vertices, three of which correspond to the literals in $c_j$; these vertices are also incident to an interconnecting edge.
    
    Due to the properties of \textsc{(2,2)-e3-sat}, $x_i$ occurs twice in unnegated form, say in clauses $c_j$ and~$c_k$ of~$B$. Its first appearance, as the $r$th literal of $c_j$ ($1\leq r \leq 3$), is represented by the interconnecting edge between vertex $v_i^1$ in the variable gadget corresponding to $x_i$ and vertex $a_j^r$ in the clause gadget corresponding to~$c_j$. Similarly the second occurrence of $x_i$, say as the $s$th literal of $c_k$ ($1\leq s\leq 3$) is represented by the interconnecting edge between $v_i^3$ and~$a_k^s$. The same variable $x_i$ also appears twice in negated form. Appropriate $a$-vertices in the gadgets representing those clauses are connected to $v_i^2$ and~$v_i^4$. We remark that this construction involves a gadget similar to one presented by Bir\'o et al.~\cite{BMM12} in their proof of the $\NP$-hardness of {\sc min bp 3-sri}.
		
		Now we prove that there is a truth assignment satisfying $B$ if and only if there is a stable matching $M$ in~$\mathcal I$.
		 
	\begin{cl}
		 For any truth assignment satisfying $B$, a stable matching $M$ can be constructed in~$\mathcal I$.
	\end{cl}
	
	\begin{proof} In Figure~\ref{fi:srtimatchings}, we define two matchings, $M_i^T$ and $M_i^F$, on the variable gadgets and three matchings, $M_j^1, M_j^2$ and $M_j^3$, on the clause gadgets.
	
	\begin{figure}[t]
		\hspace{-2mm}
		\begin{minipage}[t]{0.25\textwidth}
			\vspace{6mm}
            \def\arraystretch{1.7}
			\begin{tabular}{cc}
				$M_i^T$ \hspace{-5mm} &$= \{v_i^1v_i^2, v_i^3v_i^4\}$\\\vspace{1mm}
				$M_i^F$ \hspace{-5mm} &$= \{v_i^1v_i^4, v_i^2v_i^3\}$
			\end{tabular}
		\end{minipage}\hspace{3mm}
		\begin{minipage}[t]{0.65\textwidth}\vspace{3mm}
        \def\arraystretch{1.7}
			\begin{tabular}{cc}
				$M_j^1$ \hspace{-5mm} &$= \{a_j^1q_j^1, b_j^1p_j^1, a_j^2b_j^2, a_j^3b_j^3, q_j^2q_j^3,p_j^2p_j^3, y_j^1y_j^2, y_j^3y_j^4, z_j^1z_j^2, z_j^3z_j^4 \}$\\
				$M_j^2$ \hspace{-5mm} &$= \{a_j^2q_j^1, b_j^2p_j^1, a_j^1b_j^1, a_j^3b_j^3,q_j^2q_j^3,p_j^2p_j^3, y_j^1y_j^2, y_j^3y_j^4, z_j^1z_j^2, z_j^3z_j^4  \}$\\
				$M_j^3$ \hspace{-5mm} &$= \{a_j^3q_j^3, b_j^3p_j^3, a_j^1b_j^1, a_j^2b_j^2, q_j^1q_j^2,p_j^1p_j^2, y_j^1y_j^2, y_j^3y_j^4, z_j^1z_j^2, z_j^3z_j^4  \}$
			\end{tabular}
		\end{minipage}
		\caption{The matchings corresponding to variable $x_i$ if it is set to be $\true$ and $\false$, respectively, and to the first, second or third literal being $\true$ in a fixed clause~$c_j$.}
		\label{fi:srtimatchings}
	\end{figure}

	If a variable $x_i$ $(1\leq i\leq n)$ is assigned to be $\true$, $M_i^T$ is added to $M$, otherwise $M_i^F$ is added. Similarly, since at least one literal in $c_j$ $(1\leq j\leq m)$ is true, let $r$ ($1\leq r\leq 3)$ be the minimum integer such that the literal at position $r$ of $c_j$ is $\true$; add $M_j^r$ to~$M$. The intuition behind this choice is that if a literal is $\true$, then the vertex representing it in the variable gadget is matched to its best choice. On the other hand, if some literals in a clause are $\true$, then the vertex representing the appearance of one of them in that clause is matched to its last-choice vertex.
	
	We claim that no edge blocks~$M$. Checking the edges in the clause and variable gadgets is easy. The five special matchings were designed in such a way that no edge within the gadgets blocks them. More explanation is needed regarding the interconnecting edges. Suppose one of them, $a_j^r v_i^s$, $(r \in \left\{ 1,2,3\right\}, s \in \left\{ 1,2,3,4\right\})$ blocks~$M$. Since $M$ is a perfect matching, $a_j^r$ needs to be matched to its last choice, a $q$-vertex. Similarly, $v_i^s$ has to be matched to its worst partner. While the partner of $a_j^r$ indicates that the literal represented by $v_i^s$ ($x_i$ or $\bar{x}_i$) is $\true$ in the clause, the partner of $v_i^s$ means that the literal is $\false$.
	\end{proof}
		
	\begin{cl}
		For any stable matching $M$ in $\mathcal I$, there is a truth assignment satisfying~$B$.
	\end{cl}
	\begin{proof} In the next three paragraphs we show that the restriction of $M$ to any variable or clause gadget is one of the above listed special matchings, and no interconnecting edge is in~$M$.
	
	First of all, if a vertex $u$ is the only first choice of another vertex, then $u$ certainly needs to be matched in~$M$. This property is fulfilled for all vertices of all clause gadgets except for $y_j^3$ and $z_j^3$ for each $c_j$ ($1 \leq j \leq m$). Let us first study clause gadget~$c_j$. If $y_j^4$ is matched to $y_j^2$, then $y_j^2y_j^3$ blocks~$M$. Thus, $y_j^3 y_j^4$, and similarly, $z_j^3z_j^4$ are part of $M$ for all clause gadgets.
	
	Our proof for clause gadgets from this point involves considering matchings covering all twelve remaining vertices. We differentiate two possible cases, depending on the partner of~$p_j^3$. In the first case, $p_j^3 b_j^3 \in M$. Therefore, $p_j^2p_j^1 \in M$ too, because $p_j^2$ has to be matched. For similar reasons, $\{ b_j^1a_j^1, b_j^2a_j^2, q_j^1q_j^2, q_j^3a_j^3\}$ $\subseteq M$. This gives us matching~$M_j^3$. In the second case, if $p_j^3$ is matched to $p_j^2$, then $\{ b_j^3a_j^3, q_j^3q_j^2 \} \subseteq M$. There are two possible matchings on the remaining six vertices: $\{p_j^1b_j^1, a_j^1q_j^1, b_j^2a_j^2\}$ and $\{p_j^1b_j^2, q_j^1 a_j^2,b_j^1a_j^1\}$. These two matchings together with the lower part of the gadget form $M_j^1$ and~$M_j^2$.
	
	Since all $a$-vertices have a partner within their clause gadgets, no interconnecting edge can be a part of~$M$. For the variable gadgets, it is straightforward to see that $M_i^T$ and $M_i^F$ are the only matchings covering all vertices of the 4-cycles. 
	
	The truth assignment to $B$ is then defined in the following way. Each variable whose gadget has the edges of $M_i^T$ in $M$ is assigned to be $\true$, while all other variables with $M_i^F$ on their gadgets are $\false$.
	
	All that remains is to show that this is indeed a truth assignment. Suppose that there is an unsatisfied clause $c_j$ in~$B$. Since all three of $c_j$'s literals are $\false$, every vertex $v_i^r$ ($1\leq i\leq n$) such that $v_i^r a_j^s$ is an interconnecting edge prefers $a_j^s$ to its partner in $M$ ($1\leq s\leq 3$).  Hence a blocking edge can only be avoided if $a_j^1b_j^1$, $a_j^2b_j^2$ and $a_j^3b_j^3$ are all in~$M$, which never occurs in any stable matching as shown above.
	\end{proof}
    This finishes the proof of Theorem~\ref{th:3-srti}.
\end{proof}

Our construction shows that the complexity result holds even if the preference lists are either strictly ordered or consist of a single tie of length two. Moreover, Theorem \ref{th:3-srti} also immediately implies the following result.

\begin{cor}
	\textsc{min bp 3-srti} is $\NP$-hard.
    \label{cor:minbp3srti}
\end{cor}

The following result strengthens Corollary~\ref{cor:minbp3srti}.

\begin{theorem}
\label{th:3srti_inappr}
\textsc{min bp 3-srti} is not approximable within $n^{1-\varepsilon}$, for any $\varepsilon >0$, unless $\P=\NP$, where $n$ is the number of agents.
\end{theorem}


\begin{proof}
	The core idea of the proof is to gather several copies of the {\sc 3-srti} instance created in the proof of Theorem~\ref{th:3-srti}, together with a small unsolvable {\sc 3-srti} instance. By doing so, we create a \textsc{min bp 3-srti} instance $\mathcal{I}$ in which $bp(\mathcal{I})$ is large if the Boolean formula $B$ (originally given as an instance of \textsc{(2,2)-e3-sat}) is not satisfiable, and $bp(\mathcal{I}) = 1$ otherwise. Therefore, finding a good approximation for $\mathcal I$ will imply a polynomial-time algorithm to decide the satisfiability of~$B$. Our proof is similar to that of an analogous inapproximabilty result for the problem of finding a most-stable matching in an instance of the Hospitals / Residents problem with Couples~\cite{BMM14}. 
	
	The smallest unsolvable instance of {\sc 3-srti} is a 3-cycle with cyclic strict preferences. Aside from this, we add $k$ disjoint copies of {\sc 3-srti} instance created in the proof of Theorem \ref{th:3-srti} (from the same Boolean formula~$B$), for large enough~$k$.  In particular we let $c = \lceil2/\varepsilon\rceil$ and $k = n_0^c$, where $n_0$ is the number of variables in~$B$. We use $m_0$ to denote the number of clauses in~$B$.  Let $\mathcal I$ be the instance of {\sc 3-srti} that has been constructed.  Due to the proof of Theorem~\ref{th:3-srti} above, if $B$ is satisfiable then $bp(\mathcal I)=1$, and if $B$ is not satisfiable then $bp(\mathcal I)\geq k+1$.  Hence a $k$-approximation algorithm for {\sc min bp 3-srti} could be used to solve \textsc{(2,2)-e3-sat} in polynomial time.
	
	In the remainder of the proof we show that $n^{1-\varepsilon} \leq k$, where $n$ is the number of agents in $\mathcal I$, which will imply the statement of the theorem.  With Inequalities~(\ref{sreq:1})-(\ref{sreq:4}) we give an upper bound for~$n$. This is used in Inequalities~(\ref{sreq:5})-(\ref{sreq:8}) as we establish $k$ as an upper bound for $n^{1-\varepsilon}$. Explanations for the steps are given as and when it is necessary after each set of inequalities.
	\begin{align}
	\label{sreq:1} n &= k(4n_0+20m_0)+3 \\
	\label{sreq:2} &= k(4n_0+20\frac{4n_0}{3})+3 \\ 
	\label{sreq:3}   &\leq 32kn_0 \\
	\label{sreq:4}  &= 32 n_0^{c+1}
	\end{align}

	In Equality~(\ref{sreq:1}) can be deduced by inspection of the {\sc 3-srti} instance constructed in the proof of Theorem \ref{th:3-srti}. In step~(\ref{sreq:2}) we substitute $m_0=\frac{4n_0}{3}$, which follows from the structure of~$B$. We can assume without loss of generality that $kn_0\geq 3$, which we use in Inequality~(\ref{sreq:3}). Finally, in Equality~(\ref{sreq:4}) we substitute $k = n_0^c$.
	
	Since $c=\lceil2/\varepsilon\rceil$, the following inequality also holds. 
	\begin{align}
		\label{sreq:extra} \frac{c-1}{c+1} &= 1- \frac{2}{c+1} \geq 1-\varepsilon
	\end{align}

	We can now establish the desired upper bound for $n^{1-\varepsilon}$.
	\begin{align}
	\label{sreq:5} n^{1-\varepsilon} &\leq n^{\frac{c-1}{c+1}}\\
	\label{sreq:6} & \leq 32^{\frac{c-1}{c+1}} n_0^{c-1} \\
	\label{sreq:7} & \leq n_0^c \\
	\label{sreq:8} & =k
	\end{align}
	
	\noindent
	Inequality~(\ref{sreq:5}) is obtained by raising $n$ to the power of each side of Inequality~(\ref{sreq:extra}). Inequality~(\ref{sreq:6}) follows from the bound for $n$ established in Inequalities~(\ref{sreq:1})-(\ref{sreq:4}). Now in Inequality~(\ref{sreq:7}) we can assume without loss of generality that $n_0\geq 32$ and use that $\frac{c-1}{c+1} < 1$. In the last step, we use the definition of~$k$.
	\end{proof}

To complete the study of cases of \textsc{min bp $d$-srti}, we establish a positive result for instances with degree at most~2.

\begin{theorem}
	\textsc{min bp 2-srti} is solvable in $\mathcal{O}(|V|)$ time.
\end{theorem}

\begin{proof}
For an instance $\mathcal I$ of \textsc{min bp 2-srti}, clearly every component of the underlying graph $G$ is a path or cycle.  We claim that $bp(\mathcal{I})$ equals the number of \emph{odd parties} in~$G$, where an \emph{odd party} is a cycle $C = \langle v_1, v_2, ..., v_k \rangle$ of odd length, such that $v_i$ strictly prefers $v_{i+1}$ to $v_{i-1}$ (addition and subtraction are taken modulo~$k$).
	
	Since an odd party never admits a stable matching, $bp(\mathcal{I})$ is bounded below by the number of odd parties~\cite{Tan91}. This bound is tight: by taking an arbitrary maximum matching in an odd party component, a most-stable matching is already reached. Now we show that a stable matching $M$ can be constructed in all other components.
	
	Each component that is not an odd cycle is therefore a bipartite subgraph (indeed either a path or an even cycle). Such a subgraph therefore gives rise to the restriction of {\sc srti} called the \emph{Stable Marriage problem with Ties and Incomplete lists} ({\sc smti}).  An instance of {\sc smti} always admits a stable solution and it can be found in linear time~\cite{MIIMM02}. Thus these components contribute no blocking edge.
    
    Regarding odd-length cycles that are not odd parties, we will show that there is at least one vertex not strictly preferred by either of its adjacent vertices. Leaving this vertex uncovered and adding a perfect matching in the rest of the cycle results in a stable matching.
		
		Assume that every vertex along a cycle $C_k$ (where $k$ is an odd number) is strictly preferred by at least one of its neighbours. Since each of the $k$ vertices is strictly preferred by at least one vertex, and a vertex $v$ can prefer at most one other vertex strictly, every vertex along $C_k$ has a strictly ordered preference list. Now every vertex can point at its unique first-choice neighbour. To avoid an odd cycle, there must be a vertex pointed at by both of its neighbours. This implies that there is also a vertex $v$ pointed at by no neighbour, and $v$ is hence ranked second by both of its neighbours. 
		\end{proof}
		
\section{Egalitarian stable matchings in {\sc srti}}
\label{se:no_egal_srti}

In this section we outline the difficulties one encounters when attempting to define and study the concept of an egalitarian stable matching in instances of {\sc srti}.

\begin{itemize}
\item When considering the approximability of \textsc{egal $d$-sri}, we restricted attention to the case of solvable instances, in the knowledge that solvability can be determined in linear time~\cite{Irv85}. However in the case of {\sc srti}, we can no longer assume this, since {\sc solvable 3-srti} is $\NP$-complete as Theorem~\ref{th:3-srti} shows.
\item In instances of \textsc{egal $d$-sri}, not all agents are necessarily matched in all stable matchings, but due to Theorem 4.5.2 of~\cite{GI89}, which states that the same agents are matched in all stable matchings, we can discard unmatched agents and consider only the remaining agents when reasoning about approximation algorithms. There is no analogue of Theorem 4.5.2 in the case of \textsc{$d$-srti} (indeed, stable matchings can be of different sizes in a given instance of {\sc srti}~\cite{IM02}). This means that any approximation algorithm for the problem of finding an egalitarian stable matching in an instance of {\sc srti} would need to consider the cost of an unmatched agent in a given stable matching, and the choice of value for such a case is not universally agreed upon in the literature. Chen et al.~\cite{CHSY17} study the fixed-parameter tractability of {\sc egal srti} under different choices of cost value for an unmatched agent, namely 0, some positive constant and the length of its preference list.
\item Similarly in the case of {\sc srti}, the choice of value for the rank of an agent $a_j$ in a given agent $a_i$'s preference list is again not universally agreed upon -- for example if $a_i$ has a tie of length 2 at the head of his preference list, followed strictly by $a_j$, then $\rank(a_i,a_j)$ could reasonably be defined to be either 2 or 3 depending on the definition adopted. In most competitions, everybody in the tie receives the rank that directly follows the number of agents ranked strictly higher than them, which would be 3 in the previous example. On the other hand, setting the rank to the number of ties (of any cardinality) in the list up to the current tie is the correct way of dealing with this issue in markets where agents rank their possible partners into well-separated tiers and the cardinalities of these do not matter as much as the tier they end up being matched to -- this principle assigns 2 to $\rank(a_i,a_j)$ in the example above.
\end{itemize}

\section{Open questions}
\label{sec:open}
Theorems~\ref{th:egal3}, \ref{th:egal4} and \ref{th:egal5} improve on the best known approximation factor for \textsc{egal $d$-sri} for small~$d$. It remains open to come up with an even better approximation or to establish an inapproximability bound matching our algorithm's guarantee. A more general direction is to investigate whether the problem of finding a minimum weight stable matching can be approximated within a factor less than 2 for instances of $d$-{\sc sri} for small~$d$. Finally, the various alternatives regarding the definition of an egalitarian stable matching in instances of {\sc srti} open the gate to a number of questions.

\section*{Acknowledgements}
We thank the anonymous reviewers of this paper and an earlier version of it for their valuable comments, which helped to improve the presentation.



\end{document}